\documentclass[journal]{IEEEtran}

\usepackage[utf8]{inputenc}
\usepackage{color}
\usepackage{array}
\usepackage{verbatim}
\usepackage{float}
\usepackage{amsmath}
\usepackage{amsthm}
\usepackage{amssymb}
\usepackage{graphicx}
\usepackage{longtable}
\usepackage{multirow}
\usepackage{balance}
\usepackage[unicode=true,
bookmarks=false,
breaklinks=false,pdfborder={0 0 1},colorlinks=false]
{hyperref}
\hypersetup{
	colorlinks,bookmarksopen,bookmarksnumbered,citecolor=blue,urlcolor=blue}
\usepackage{cite}

\usepackage{lipsum}
\usepackage{mathtools}
\usepackage{cuted}

\floatstyle{ruled}
\newfloat{algorithm}{tbp}{loa}
\providecommand{\algorithmname}{Algorithm}
\floatname{algorithm}{\protect\algorithmname}

\makeatletter
\let\oldforeign@language\foreign@language
\DeclareRobustCommand{\foreign@language}[1]{%
	\lowercase{\oldforeign@language{#1}}}

\let\oldforeign@language\foreign@language
\DeclareRobustCommand{\foreign@language}[1]{%
	\lowercase{\oldforeign@language{#1}}}

\ifCLASSINFOpdf
\else
\fi

\newcommand{\MYfooter}{\smash{
		\hfil\parbox[t][\height][t]{\textwidth}{\centering
			\thepage}\hfil\hbox{}}}

\def\ps@IEEEtitlepagestyle{%
	\def\@oddhead{\parbox[t][\height][t]{\textwidth}{\centering \scriptsize
			Personal use of this material is permitted. Permission from the author(s) and/or copyright holder(s), must be obtained for all other uses. Please contact us and provide details if you believe this document breaches copyrights.\\
			\noindent\makebox[\linewidth]{}
		}\hfil\hbox{}}%
	\def\@evenhead{\scriptsize\thepage \hfil \leftmark\mbox{}}%
	\def\@oddfoot{\parbox[t][\height][l]{\textwidth}{
			\vspace{-20pt}{\rule{\textwidth}{0.4pt}}\\ \footnotesize\underline{To cite this article:}
			{\bf{\footnotesize\textcolor{red}{H. A. Hashim, M. Abouheaf, K. G. Vamvoudakis, "Neural-adaptive Stochastic Attitude Filter on SO(3)," IEEE Control Systems Letters, vol. 6, no. 1, pp. 1549-1554, 2022.}}} doi: \href{https://doi.org/10.1109/LCSYS.2021.3123227}{10.1109/LCSYS.2021.3123227}\\
			\noindent\makebox[\linewidth]
		}\hfil\hbox{}}%
	\def\@evenfoot{\MYfooter}}

\makeatother
\newtheorem{defn}{Definition}

\newtheorem{lem}{Lemma}

\newtheorem{thm}{Theorem}

\newtheorem{assum}{Assumption}

\begin{document}
	\bstctlcite{IEEEexample:BSTcontrol}

	\title{Neural-adaptive Stochastic Attitude Filter on SO(3)}

\author{Hashim A. Hashim, Mohammed Abouheaf, and Kyriakos G. Vamvoudakis
		\thanks{This work was supported in part by the Dean’s Office and the Department of Mechanical and Aerospace Engineering through the Carleton University Startup Research Grant, and in part by the NSF under grant Nos. CAREER CPS-1851588 and S\&AS 1849198.}
		\thanks{H. A. Hashim is with the Department of Mechanical and Aerospace Engineering, Carleton University, Ottawa, ON, K1S 5B6, Canada (e-mail: HashimMohamed@cunet.carleton.ca), M. Abouheaf is with College of Technology, Architecture \& Applied Engineering, Bowling Green State University, Bowling Green, 43402, OH, USA, (email: mabouhe@bgsu.edu), and K. G. Vamvoudakis is with the Daniel Guggenheim School of Aerospace Engineering, Georgia Institute of Technology, Atlanta, GA, 30332, USA (e-mail: kyriakos@gatech.edu).}
}



\maketitle

\begin{abstract}
Successful control of a rigid-body rotating in three dimensional space
requires accurate estimation of its attitude. The attitude dynamics
are highly nonlinear and are posed on the Special Orthogonal Group
$SO(3)$. In addition, measurements supplied by low-cost sensing units
pose a challenge for the estimation process. This paper proposes a
novel stochastic nonlinear neural-adaptive-based filter on $SO(3)$
for the attitude estimation problem. The proposed filter produces
good results given measurements extracted from low-cost sensing units (e.g., IMU or MARG sensor modules). The filter is guaranteed to be almost semi-globally uniformly ultimately
bounded in the mean square. In addition to Lie Group formulation,
quaternion representation of the proposed filter is provided. The
effectiveness of the proposed neural-adaptive filter is tested and
evaluated in its discrete form under the conditions of large initialization
error and high measurement uncertainties.
\end{abstract}


\IEEEpeerreviewmaketitle{}

\section{Introduction}

	\IEEEPARstart{R}{obotics} and control applications are heavily reliant on robust filtering
	solutions to guarantee feasibility of accurate rigid-body orientation
	(attitude) estimation \cite{markley1988attitude,markley2003attitude,crassidis2003unscented,zlotnik2016exponential}.
	The attitude can be reconstructed algebraically given known observations
	in the inertial-frame and the associated measurements in the body-frame.
	Examples include QUEST algorithm \cite{shuster1981three} and singular
	value decomposition (SVD) \cite{markley1988attitude}. However, body-frame
	measurements might be attached with uncertainties, in particular if
	they were supplied by low-cost inertial measurement units (IMUs) or
	magnetic, angular rate, and gravity (MARG) sensor. Hence, accounting
	for measurement imperfections requires substituting algebraic attitude
	reconstruction with estimation filters.
	
	The problem of attitude estimation is traditionally tackled by the
	active control and robotics research community using Gaussian filters,
	such as, Kalman filter (KF) \cite{choukroun2006novel}, extended Kalman
	filter (EKF) \cite{madyastha2011extended}, multiplicative extended
	Kalman filter (MEKF) \cite{markley2003attitude}, unscented Kalman
	filter (UKF) \cite{crassidis2003unscented}, and invariant extended
	Kalman filter (IEKF) \cite{bonnabel2007left}. The unit-quaternion
	structure of the majority of Gaussian filters offers the benefit of
	nonsingular attitude representation \cite{hashim2018SO3Stochastic,hashim2019SO3Wiley}. However,
	on the other hand, unit-quaternion formulation is subject to nonuniqueness
	\cite{shuster1993survey,hashim2019AtiitudeSurvey}. This motivated the researchers to explore
	posing the attitude on the Special Orthogonal Group $\mathbb{SO}(3)$.
	Unlike unit-quaternion, $\mathbb{SO}(3)$ offers unique and global
	representation of the rotational matrix \cite{zlotnik2016exponential,hashim2019SO3Det,mahony2008nonlinear,lee2018bayesian,hashim2018SO3Stochastic,hashim2019SO3Wiley}.
	Therefore, over the last decade multiple nonlinear attitude filters
	on $\mathbb{SO}(3)$ have been proposed, such as nonlinear deterministic
	filters \cite{zlotnik2016exponential,hashim2019SO3Det,mahony2008nonlinear,lee2018bayesian}
	and nonlinear stochastic filters \cite{hashim2018SO3Stochastic,hashim2019SO3Wiley}.
	The nonlinear filter design on $\mathbb{SO}(3)$ has proven to 1)
	have a simpler structure, 2) be computationally cheap, and 3) have
	better tracking performance in contrast to Gaussian filters \cite{zlotnik2016exponential,hashim2019SO3Det,mahony2008nonlinear,lee2018bayesian,hashim2018SO3Stochastic,hashim2019SO3Wiley}.
	
	It is widely known that neural networks (NNs) have capability to learn
	complex nonlinear relationships \cite{zhao2018neuroadaptive,wang2016fraction,song2017indirect,song2018neuroadaptive}.
	In the recent years, adaptive artificial neural networks (ANNs) learning,
	known as neural-adaptive learning, has been found effective for approximating
	unknown nonlinear dynamics online in several control applications.
	Examples include two-degrees-of-freedom arm robots \cite{zhao2018neuroadaptive},
	multi-agent systems \cite{wang2016fraction}, unknown multi-input
	multi-output systems \cite{song2017indirect} and fault-tolerant control
	\cite{song2018neuroadaptive}. Accurate NN approximation of unknown
	nonlinear dynamics allows for successful control process \cite{zhao2018neuroadaptive,wang2016fraction,song2017indirect,song2018neuroadaptive}.
	In this work, the attitude dynamics are modelled on the Lie Group
	of $\mathbb{SO}(3)$. The uncertainties inherent to attitude dynamics
	and gyroscope measurements, are addressed using Brownian motion process.
	The contributions of this paper are as follows: 1) a neural-adaptive
	nonlinear stochastic attitude filter on $\mathbb{SO}(3)$ is proposed,
	2) the measurement uncertainties are corrected using neural-adaptive
	adaptation mechanisms extracted by adopting Lyapunov stability, and
	3) the closed loop signals are guaranteed to be almost semi-globally
	uniformly ultimately bounded (SGUUB). While the filter is proposed
	in a continuous form, its discrete form obtained using exact integration
	methods is also presented. The filter is tested at a low sampling
	rate to reflect real-life applications. To the best of the authors
	knowledge, the attitude estimation problem has not been addressed
	using a neural-adaptive stochastic filter on $\mathbb{SO}(3)$.
	
	The paper is structured to include six Sections. Section \ref{sec:Preliminaries}
	presents preliminaries of the attitude problem. Section \ref{sec:SE3_Problem-Formulation}
	defines the problem, contains the available measurements, error criteria,
	and neural network approximation. Section \ref{sec:Neuro-adaptive-based-Stochastic-}
	presents a novel neural-adaptive stochastic attitude filter. Section
	\ref{sec:SE3_Simulations} shows and discusses the obtained results.
	Lastly, Section \ref{sec:SE3_Conclusion} concludes the paper.
	
	\section{Preliminaries\label{sec:Preliminaries}}
	
	In this work, $\mathbb{R}$ represents the set of real numbers, $\mathbb{R}_{+}$
	denotes the set of nonnegative real numbers, and $\mathbb{R}^{n\times m}$
	stands for a real $n$-by-$m$ dimensional space. $\mathbf{I}_{n}$
	and $0_{n\times m}$ denote an $n$-by-$n$ identity matrix and an
	$n$-by-$m$ dimensional matrix of zeros, respectively. For $a\in\mathbb{R}^{n}$
	and $A\in\mathbb{R}^{n\times m}$, $||a||=\sqrt{a^{\top}a}$ stands
	for Euclidean norm of $x$ and $||A||_{F}=\sqrt{{\rm Tr}\{AA^{*}\}}$
	describes the Frobenius norm of $A$ where $*$ denotes a conjugate
	transpose. For $A\in\mathbb{R}^{n\times n}$, define a set of eigenvalues
	as $\lambda(A)=\{\lambda_{1},\lambda_{2},\ldots,\lambda_{n}\}$ where
	$\overline{\lambda}_{A}=\overline{\lambda}(A)$ denotes the maximum
	value, while $\underline{\lambda}_{A}=\underline{\lambda}(A)$ describes
	the minimum value of $\lambda(A)$. $\left\{ \mathcal{I}\right\} $
	defines a fixed inertial-frame and $\left\{ \mathcal{B}\right\} $
	describes a fixed body-frame. Rigid-body's orientation in three-dimensional
	space, commonly known as attitude, is expressed as $R\in\mathbb{SO}(3)$
	with
	\[
	\mathbb{SO}(3)=\{R\in\mathbb{R}^{3\times3}|R^{\top}R=\mathbf{I}_{3}\text{, }{\rm det}(R)=+1\}
	\]
	where ${\rm det}(\cdot)$ denotes a determinant. The Lie algebra associated
	with $\mathbb{SO}(3)$ is termed $\mathfrak{so}(3)$ and can be described
	as
	\begin{align*}
	\mathfrak{so}(3) & =\{[a]_{\times}\in\mathbb{R}^{3\times3}|[a]_{\times}^{\top}=-[a]_{\times},a\in\mathbb{R}^{3}\}\\
	\left[a\right]_{\times} & =\left[\begin{array}{ccc}
	0 & -a_{3} & a_{2}\\
	a_{3} & 0 & -a_{1}\\
	-a_{2} & a_{1} & 0
	\end{array}\right]\in\mathfrak{so}\left(3\right),\hspace{1em}a=\left[\begin{array}{c}
	a_{1}\\
	a_{2}\\
	a_{3}
	\end{array}\right]
	\end{align*}
	The operator $\mathbf{vex}$ stands for the inverse mapping of $[\cdot]_{\times}$
	with the map $\mathbf{vex}:\mathfrak{so}(3)\rightarrow\mathbb{R}^{3}$
	where $\mathbf{vex}([a]_{\times})=a,\forall a\in\mathbb{R}^{3}$.
	The anti-symmetric projection has the map $\boldsymbol{\mathcal{P}}_{a}:\mathbb{R}^{3\times3}\rightarrow\mathfrak{so}(3)$
	where
	\[
	\boldsymbol{\mathcal{P}}_{a}(M)=\frac{1}{2}(M-M^{\top})\in\mathfrak{so}\left(3\right),\forall M\in\mathbb{R}^{3\times3}
	\]
	For $M=[m_{i,j}]_{i,j=1,2,3}\in\mathbb{R}^{3\times3}$, let us define
	\begin{equation}
	\boldsymbol{\Upsilon}(M)=\mathbf{vex}(\boldsymbol{\mathcal{P}}_{a}(M))=\frac{1}{2}\left[\begin{array}{c}
	m_{32}-m_{23}\\
	m_{13}-m_{31}\\
	m_{21}-m_{12}
	\end{array}\right]\in\mathbb{R}^{3}\label{eq:Attit_VEX}
	\end{equation}
	For $R\in\mathbb{SO}\left(3\right)$, define the Euclidean distance
	of $R$ as follows:
	\begin{equation}
	||R||_{{\rm I}}=\frac{1}{4}{\rm Tr}\{\mathbf{I}_{3}-R\}\in\left[0,1\right]\label{eq:Attit_Ecul_Dist}
	\end{equation}
	with ${\rm Tr}\{\cdot\}$ standing for a trace of a matrix. For $A\in\mathbb{R}^{3\times3}$
	and $\alpha\in\mathbb{R}^{3}$, considering the composition mapping
	in \eqref{eq:Attit_VEX}, let us introduce the following identity:
	\begin{equation}
	{\rm Tr}\{A[\alpha]_{\times}\}={\rm Tr}\{\boldsymbol{\mathcal{P}}_{a}(A)[\alpha]_{\times}\}=-2\boldsymbol{\Upsilon}(A)^{\top}\alpha\label{eq:Attit_Identity1}
	\end{equation}

	\section{Problem Formulation\label{sec:SE3_Problem-Formulation}}
	
	\subsection{Measurements and Dynamics}
	
	Let $R\in\mathbb{SO}(3)$ be the attitude of a rigid-body in three-dimensional
	space defined with respect to $\left\{ \mathcal{B}\right\} $. The
	true attitude dynamics:
	\begin{equation}
	\dot{R}=R\left[\Omega\right]_{\times}\label{eq:Attit_R_dot}
	\end{equation}
	where $\Omega\in\mathbb{R}^{3}$ represents angular velocity of the
	rigid-body defined with respect to $\left\{ \mathcal{B}\right\} $.
	The attitude of a rigid-body can be obtained given a group of measurements
	in $\left\{ \mathcal{B}\right\} $ and a group of observations in
	$\left\{ \mathcal{I}\right\} $. Let $r_{i}\in\mathbb{R}^{3}$ denote
	an observation in $\left\{ \mathcal{I}\right\} $. As such, the measurement
	of $r_{i}$ with respect to $\left\{ \mathcal{B}\right\} $ is given
	by
	\begin{align}
	y_{i} & =R^{\top}r_{i}+n_{i}\in\mathbb{R}^{3},\hspace{1em}\forall i=1,2,\ldots,N\label{eq:Attit_Vec_yi}
	\end{align}
	where $n_{i}$ denotes unknown noise. The attitude can be obtained
	given two or more non-collinear inertial observations ($N\geq2$)
	and the respective body-frame measurements. If $N=2$, the third observation
	and the associated measurement can be defined by $r_{3}=r_{2}\times r_{1}$
	and $y_{3}=y_{2}\times y_{1}$ where $\times$ denotes a cross product.
	The set of observations and measurements can be normalized as follows:
	\begin{equation}
	{\bf r}_{i}=\frac{r_{i}}{||r_{i}||},\hspace{1em}{\bf y}_{i}=\frac{y_{i}}{||y_{i}||}\label{eq:Attit_Vec_norm}
	\end{equation}
	Low-cost IMU or MARG sensors can be utilized for attitude determination
	or estimation, see \cite{zlotnik2016exponential,hashim2019SO3Det,mahony2008nonlinear,lee2018bayesian,hashim2018SO3Stochastic,hashim2019SO3Wiley}.
	Gyroscope (angular rate or angular velocity) measurements can be defined
	as follows:
	\begin{equation}
	\Omega_{m}=\Omega+n\in\mathbb{R}^{3}\label{eq:Attit_Om_m}
	\end{equation}
	with $\Omega$ being the true angular velocity defined in \eqref{eq:Attit_R_dot},
	and $n$ being unknown noise corrupting $\Omega_{m}$. %
	The noise vector $n$ is bounded and Gaussian with a zero mean $\mathbb{E}[n]=0$
	where $\mathbb{E}[\cdot]$ denotes expected value of a component.
	Derivative of a Gaussian process results in a Gaussian process \cite{Hashim2021AESCTE,khasminskii1980stochastic}.
	As such, $n$ can be formulated as a Brownian motion process
	\begin{equation}
	n=\mathcal{Q}\frac{d\beta}{dt}\label{eq:Attit_n_beta}
	\end{equation}
	where $\beta\in\mathbb{R}^{3}$ and $\mathcal{Q}\in\mathbb{R}^{3\times3}$
	is an unknown time-variant symmetric matrix with $\mathcal{Q}^{2}=\mathcal{Q}\mathcal{Q}^{\top}$
	being the noise covariance. It is worth noting that $\mathbb{P}\{\beta(0)=0\}=1$
	and $\mathbb{E}[\beta]=0$ where $\mathbb{P}\{\cdot\}$ denotes probability
	of a component. Therefore, from \eqref{eq:Attit_R_dot}, \eqref{eq:Attit_Om_m},
	and \eqref{eq:Attit_n_beta}, the true attitude dynamics can be defined
	in a stochastic sense as follows:
	\begin{equation}
	dR=R[\Omega_{m}]_{\times}dt-R[\mathcal{Q}d\beta]_{\times}\label{eq:Attit_dR_dt}
	\end{equation}
	In view of \eqref{eq:Attit_VEX}-\eqref{eq:Attit_Identity1}, one
	obtains the normalized Euclidean distance of $R$ in \eqref{eq:Attit_dR_dt}
	as follows:
	\begin{equation}
	d||R||_{{\rm I}}=2\boldsymbol{\Upsilon}(R)^{\top}\Omega_{m}dt-2\boldsymbol{\Upsilon}(R)^{\top}\mathcal{Q}d\beta\label{eq:Attit_dR_norm}
	\end{equation}
	
	\begin{lem}
		\label{lem:Lem1}\cite{hashim2019SO3Wiley} Let $R\in\mathbb{SO}(3)$,
		$\boldsymbol{\Upsilon}(R)=\mathbf{vex}(\boldsymbol{\mathcal{P}}_{a}(R))$
		as in \eqref{eq:Attit_VEX}, and $||R||_{{\rm I}}=\frac{1}{4}{\rm Tr}\{\mathbf{I}_{3}-R\}$
		as \eqref{eq:Attit_Ecul_Dist}. Hence, the following equality holds:
		\[
		||\boldsymbol{\Upsilon}(R)||^{2}=4(1-||R||_{{\rm I}})||R||_{{\rm I}}
		\]
	\end{lem}
	\begin{defn}
		\label{def:Def_SGUUB}\cite{hashim2019SO3Wiley,ji2006adaptive} Consider the stochastic
		attitude dynamics in \eqref{eq:Attit_dR_norm} and let $t_{0}$ be
		the initial time. $||R||_{{\rm I}}=||R(t)||_{{\rm I}}$ is said to
		be almost SGUUB if for a given set $\pi\in\mathbb{R}$ and $||R(t_{0})||_{{\rm I}}$
		a constant $\alpha>0$ exists and a time constant $T_{\alpha}=T_{\alpha}(\kappa,||R(t_{0})||_{{\rm I}})$
		such that $\mathbb{E}[||R(t_{0})||_{{\rm I}}]<\alpha,\forall t>t_{0}+\alpha$.
	\end{defn}
	\begin{lem}
		\label{Lemm:Def_LV_dot}\cite{deng2001stabilization} Recall the stochastic
		attitude dynamics in \eqref{eq:Attit_dR_norm} and assume that $V(||R||_{{\rm I}})$
		be a twice differentiable potential function such that
		\begin{equation}
		\mathcal{L}V(||R||_{{\rm I}})=V_{1}^{\top}f+\frac{1}{2}{\rm Tr}\{g\mathcal{Q}^{2}g^{\top}V_{2}\}\label{eq:Attit_Vfunction_Lyap0}
		\end{equation}
		with $f=2\boldsymbol{\Upsilon}(R)^{\top}\Omega_{m}\in\mathbb{R}$,
		$g=-2\boldsymbol{\Upsilon}(R)^{\top}\in\mathbb{R}^{1\times3}$, $\mathcal{L}V(||R||_{{\rm I}})$
		being a differential operator, $V_{1}=\partial V/\partial||R||_{{\rm I}}$,
		and $V_{2}=\partial^{2}V/\partial||R||_{{\rm I}}^{2}$. Let $\underline{\alpha}_{1}(\cdot)$
		and $\overline{\alpha}_{2}(\cdot)$ be class $\mathcal{K}_{\infty}$
		functions, and assume that the constants $\beta>0$ and $\eta\geq0$
		such that
		\begin{align}
		& \hspace{1em}\underline{\alpha}_{1}(||R||_{{\rm I}})\leq V(||R||_{{\rm I}})\leq\overline{\alpha}_{2}(||R||_{{\rm I}})\label{eq:Attit_Vfunction_Lyap}\\
		\mathcal{L}V(||R||_{{\rm I}}) & =V_{1}^{\top}f+\frac{1}{2}{\rm Tr}\{g\mathcal{Q}^{2}g^{\top}V_{2}\}\nonumber \\
		& \leq-\beta V(||R||_{{\rm I}})+\eta\label{eq:Attit_dVfunction_Lyap}
		\end{align}
		Hence, the stochastic attitude dynamics in \eqref{eq:Attit_dR_norm}
		have an almost unique strong solution on $[0,\infty)$. Moreover,
		the solution $||R||_{{\rm I}}$ is upper bounded in probability with
		\begin{equation}
		\mathbb{E}[V(||R||_{{\rm I}})]\leq V(||R(0)||_{{\rm I}}){\rm exp}(-\beta t)+\eta/\beta\label{eq:Attit_EVfunction_Lyap}
		\end{equation}
		Also, \eqref{eq:Attit_EVfunction_Lyap} implies that $||R||_{{\rm I}}$
		is SGUUBin the mean square.
	\end{lem}
	Define $\hat{R}$ as the estimate of $R$. Define the error in estimation
	by
	\begin{equation}
	\tilde{R}=R^{\top}\hat{R}\label{eq:Attit_Re}
	\end{equation}

	\subsection{Filter Structure and Error Dynamics}
	
	Define the filter dynamics as follows:
	\begin{equation}
	\dot{\hat{R}}=\hat{R}[\Omega_{m}-C]_{\times}\label{eq:Attit_Rest_dot}
	\end{equation}
	with $C\in\mathbb{R}^{3\times1}$ being a neural-adaptive-based correction
	matrix to be designed in the subsequent Section. From \eqref{eq:Attit_R_dot}
	and \eqref{eq:Attit_Rest_dot}, the error dynamics are as follows:
	\begin{align}
	d\tilde{R} & =R^{\top}d\hat{R}+dR^{\top}\hat{R}\nonumber \\
	& =(\tilde{R}[\Omega-C]_{\times}+[\Omega]_{\times}^{\top}\tilde{R})dt+\tilde{R}[\mathcal{Q}d\beta]_{\times}\nonumber \\
	& =\tilde{R}[\Omega]_{\times}-[\Omega]_{\times}\tilde{R}-\tilde{R}[C]_{\times}dt+\tilde{R}[\mathcal{Q}d\beta]_{\times}\label{eq:Attit_dRe}
	\end{align}
	In view of \eqref{eq:Attit_Identity1} and \eqref{eq:Attit_dRe},
	one obtains the Euclidean distance of \eqref{eq:Attit_dRe} as below:
	\begin{align}
	& d||\tilde{R}||_{{\rm I}}=d\frac{1}{4}{\rm Tr}\{\mathbf{I}_{3}-\tilde{R}\}=-\frac{1}{4}{\rm Tr}\{d\tilde{R}\}\nonumber \\
	& \hspace{2em}=\frac{1}{4}{\rm Tr}\{\tilde{R}[Cdt-\mathcal{Q}d\beta]_{\times}\}-\frac{1}{4}{\rm Tr}\{\tilde{R}[\Omega]_{\times}-[\Omega]_{\times}\tilde{R}\}\nonumber \\
	& \hspace{2em}=\frac{1}{4}{\rm Tr}\{\mathcal{P}_{a}(\tilde{R})[Cdt-\mathcal{Q}d\beta]_{\times}\}\nonumber \\
	& \hspace{2em}=-\frac{1}{2}\boldsymbol{\Upsilon}(\tilde{R})^{\top}Cdt+\frac{1}{2}\boldsymbol{\Upsilon}(\tilde{R})^{\top}\mathcal{Q}d\beta\label{eq:Attit_dRe_norm}
	\end{align}
	where ${\rm Tr}\{\tilde{R}[\Omega]_{\times}-[\Omega]_{\times}\tilde{R}\}=0$. 
	
	\subsection{Neural Network Structure}
	
	In this work, NNs with a linear in parameter structure will be employed.
	For $x\in\mathbb{R}^{n}$ and a function $f(x)\in\mathbb{R}^{m}$,
	one has
	\[
	f(x)=W^{\top}\varphi(x)+\alpha_{f}
	\]
	where $W\in\mathbb{R}^{q\times m}$ denotes a $q$-by-$m$-dimensional
	matrix of synaptic weights, $\varphi(x)\in\mathbb{R}^{q}$ denotes
	an activation function, $q$ denotes number of neurons, and $\alpha_{f}\in\mathbb{R}^{m}$
	denotes an approximated error vector. The activation function may
	contain high order connections, for instance, Gaussian functions \cite{gundogdu2016multiplicative},
	radial basis functions (RBFs) \cite{she2019battery}, sigmoid functions
	\cite{elfwing2018sigmoid}. Our objectives are to achieve accurate
	estimation of the attitude matrix, estimate the nonlinear attitude
	dynamics, and compensate for the uncertainties. NNs have been proven
	to be successful in estimating high-order nonlinear dynamics \cite{zhao2018neuroadaptive,wang2016fraction,song2017indirect,song2018neuroadaptive}.
	Recall the nonlinear dynamics in \eqref{eq:Attit_dRe_norm}
	\[
	d||\tilde{R}||_{{\rm I}}=-\frac{1}{2}\boldsymbol{\Upsilon}(\tilde{R})^{\top}Cdt+\frac{1}{2}\boldsymbol{\Upsilon}(\tilde{R})^{\top}\mathcal{Q}d\beta
	\]
	Define $\varphi(\boldsymbol{\Upsilon}(\tilde{R}))$ as an activation
	function, and let us approximate
	\begin{align*}
	C^{\top}\boldsymbol{\Upsilon}(\tilde{R}) & =C^{\top}\Gamma_{c}^{\top}\varphi(\boldsymbol{\Upsilon}(\tilde{R}))+\alpha_{b}\\
	\mathcal{Q}\boldsymbol{\Upsilon}(\tilde{R}) & =W_{\sigma}^{\top}\varphi(\boldsymbol{\Upsilon}(\tilde{R}))+\alpha_{\sigma}
	\end{align*}
	where $\varphi(\boldsymbol{\Upsilon}(\tilde{R}))\in\mathbb{R}^{q\times1}$
	is an activation function, $\Gamma_{c}\in\mathbb{R}^{q\times3}$ is
	a known weighted matrix, $C\in\mathbb{R}^{3\times1}$ is a correction
	weights vector to be adaptively tuned, $W_{\sigma}\in\mathbb{R}^{q\times3}$
	are the unknown NN weights to be adaptively tuned, $q>0$ is an integer
	that denotes the number of neurons, and $\alpha_{b}\in\mathbb{R}$
	and $\alpha_{\sigma}\in\mathbb{R}^{3}$ are the approximated error
	components. Note that $\alpha_{b},||\alpha_{\sigma}||\rightarrow0$
	as $q\rightarrow\infty$. Therefore, the error
	dynamics of the Euclidean distance in \eqref{eq:Attit_dRe_norm} can
	be reformulated as below:
	\begin{align}
	d||\tilde{R}||_{{\rm I}}=\tilde{f}dt+\tilde{g}\mathcal{Q}d\beta= & -\frac{1}{2}(C^{\top}\Gamma_{c}^{\top}\varphi(\boldsymbol{\Upsilon}(\tilde{R}))+\alpha_{b})dt\nonumber \\
	& +\frac{1}{2}(\varphi(\boldsymbol{\Upsilon}(\tilde{R}))^{\top}W_{\sigma}+\alpha_{\sigma}^{\top})d\beta\label{eq:Attit_dRe_normNN}
	\end{align}
	Define $W_{\sigma}$ as an unknown symmetric constant matrix of NN
	weights where $\overline{W}_{\sigma}=W_{\sigma}W_{\sigma}^{\top}\in\mathbb{R}^{q\times q}$.
	Let $\hat{W}_{\sigma}\in\mathbb{R}^{q\times q}$ be the estimate of
	$\overline{W}_{\sigma}$, and the error in NN weights be
	\begin{equation}
	\tilde{W}_{\sigma}=\overline{W}_{\sigma}-\hat{W}_{\sigma}\in\mathbb{R}^{q\times q}\label{eq:Attit_We}
	\end{equation}

	\section{Neural-adaptive-based Stochastic Filter Design\label{sec:Neuro-adaptive-based-Stochastic-}}
	
	In this Section, our objective is to develop a nonlinear stochastic
	filter based on neural-adaptive techniques for the attitude estimation
	problem. Consider the following neural-adaptive-based nonlinear stochastic
	filter design:
	\begin{equation}
	\begin{cases}
	\dot{\hat{R}} & =\hat{R}[\Omega_{m}-C]_{\times}\\
	\dot{\hat{W}}_{\sigma} & =\frac{\psi_{2}}{2}\Gamma_{\sigma}\varphi(\boldsymbol{\Upsilon}(\tilde{R}))\varphi(\boldsymbol{\Upsilon}(\tilde{R}))^{\top}-k_{\sigma}\Gamma_{\sigma}\hat{W}_{\sigma}\\
	C & =\left(\Gamma_{c}^{\top}+\frac{\psi_{2}}{2\psi_{1}}(\Gamma_{c}^{\top}\Gamma_{c})^{-1}\Gamma_{c}^{\top}\hat{W}_{\sigma}\right)\varphi(\boldsymbol{\Upsilon}(\tilde{R}))
	\end{cases}\label{eq:NAV_Filter1_Detailed}
	\end{equation}
	where $k_{\sigma}\in\mathbb{R}$ and $k_{c}\in\mathbb{R}$ are positive
	constants, $\Gamma_{\sigma}\in\mathbb{R}^{q\times q}$ is a positive
	diagonal matrix, $\Gamma_{c}\in\mathbb{R}^{q\times3}$ with $\Gamma_{c}^{\top}\Gamma_{c}$
	being positive definite, $q$ denotes the number of neurons, $\hat{W}_{\sigma}\in\mathbb{R}^{q\times q}$
	is the estimate of $\overline{W}_{\sigma}$, and $\tilde{R}=R_{y}^{\top}\hat{R}$
	with $R_{y}$ being the reconstructed attitude, see QUEST \cite{shuster1981three}
	or SVD \cite{markley1988attitude}. $\boldsymbol{\Upsilon}(\tilde{R})=\mathbf{vex}(\boldsymbol{\mathcal{P}}_{a}(\tilde{R}))$,
	$||\tilde{R}||_{{\rm I}}=\frac{1}{4}{\rm Tr}\{\mathbf{I}_{3}-\tilde{R}\}$,
	$\psi_{1}=\frac{1}{2}(1+||\tilde{R}||_{{\rm I}})\exp(||\tilde{R}||_{{\rm I}})$,
	and $\psi_{2}=\frac{1}{2}(2+||\tilde{R}||_{{\rm I}})\exp(||\tilde{R}||_{{\rm I}})$.
	It is becomes apparent that $\hat{W}_{\sigma}$ is symmetric for $\hat{W}_{\sigma}(0)=\hat{W}_{\sigma}(0)^{\top}$.
	It is worth noting that $\Gamma_{c}$ defines the convergence rate
	of $||\tilde{R}||_{{\rm I}}$ to the neighbourhood of the origin,
	while $\Gamma_{\sigma}$ defines the convergence rate of $\hat{W}_{\sigma}$
	to $\overline{W}_{\sigma}$. 
	\begin{thm}
		\label{thm:Theorem1}Recall the stochastic attitude dynamics in \eqref{eq:Attit_dR_dt}.
		Assume the availability of at least two observations and their respective
		measurements in \eqref{eq:Attit_Vec_yi} at each time instant. Consider
		the nonlinear neural-adaptive stochastic filter in \eqref{eq:NAV_Filter1_Detailed}
		supplied with measurements in \eqref{eq:Attit_Om_m} $\Omega_{m}=\Omega+n$
		and \eqref{eq:Attit_Vec_yi} $y_{i}=R^{\top}r_{i}$ for all $\forall i=1,2,\ldots,N$.
		Hence, for $||\tilde{R}(0)||_{{\rm I}}\neq+1$ (unstable equilibria),
		all the closed-loop errors are SGUUB in the mean square.
	\end{thm}
	\begin{proof}Let $V=V(||\tilde{R}||_{{\rm I}},\tilde{W}_{\sigma})$
		be a Lyapunov function candidate defined as
		\begin{equation}
		V=2||\tilde{R}||_{{\rm I}}\exp(||\tilde{R}||_{{\rm I}})+\frac{1}{2}{\rm Tr}\{\tilde{W}_{\sigma}^{\top}\Gamma_{\sigma}^{-1}\tilde{W}_{\sigma}\}\label{eq:Attit_Lyap}
		\end{equation}
		with the map $V:\mathbb{SO}\left(3\right)\times\mathbb{R}^{q\times q}\rightarrow\mathbb{R}_{+}$.
		Since $\exp(||\tilde{R}||_{{\rm I}})\leq\exp(1)<3$, one obtains
		\begin{align*}
		e^{\top}\underbrace{\left[\begin{array}{cc}
			1 & 0\\
			0 & \frac{1}{2}\underline{\lambda}(\Gamma_{\sigma}^{-1})
			\end{array}\right]}_{H_{1}}e\leq & V\leq e^{\top}\underbrace{\left[\begin{array}{cc}
			3 & 0\\
			0 & \frac{1}{2}\overline{\lambda}(\Gamma_{\sigma}^{-1})
			\end{array}\right]}_{H_{2}}e
		\end{align*}
		such that
		\[
		\underline{\lambda}(H_{1})||e||^{2}\leq V\leq\overline{\lambda}(H_{2})||e||^{2}
		\]
		where $e=[\sqrt{||\tilde{R}||_{{\rm I}}},||\tilde{W}_{\sigma}||_{F}]^{\top}$
		and $\underline{\lambda}(\Gamma_{\sigma}^{-1})$ and $\overline{\lambda}(\Gamma_{\sigma}^{-1})$
		stand for the minimum and the maximum eigenvalues of $\Gamma_{\sigma}^{-1}$,
		respectively. Since both $\underline{\lambda}(\Gamma_{\sigma}^{-1})>0$
		and $\overline{\lambda}(\Gamma_{\sigma}^{-1})>0$, $\underline{\lambda}(H_{1})$
		and $\overline{\lambda}(H_{2})$ are positive and $V(||\tilde{R}||_{{\rm I}},\tilde{W}_{\sigma})>0$
		for all $e\in\mathbb{R}^{2}\backslash\{0\}$. Consequently, on has
		\begin{equation}
		\begin{cases}
		\frac{\partial V}{\partial||\tilde{R}||_{{\rm I}}}= & 2\psi_{1}=2(1+||\tilde{R}||_{{\rm I}})\exp(||\tilde{R}||_{{\rm I}})\\
		\frac{\partial^{2}V}{\partial||\tilde{R}||_{{\rm I}}^{2}}= & 2\psi_{2}=2(2+||\tilde{R}||_{{\rm I}})\exp(||\tilde{R}||_{{\rm I}})
		\end{cases}\label{eq:Attit_Lvv}
		\end{equation}
		In view of \eqref{eq:Attit_Lyap}, \eqref{eq:Attit_Lvv}, and Lemma
		\ref{Lemm:Def_LV_dot}, the following differential operator is obtained:
		\begin{equation}
		\mathcal{L}V=\psi_{1}\tilde{f}+\frac{1}{2}{\rm Tr}\left\{ \tilde{g}\tilde{g}^{\top}\psi_{2}\right\} -{\rm Tr}\{\tilde{W}_{\sigma}^{\top}\Gamma_{\sigma}^{-1}\dot{\hat{W}}_{\sigma}\}\label{eq:Attit_Lyap_1}
		\end{equation}
		From \eqref{eq:NAV_Filter1_Detailed}
		\begin{align}
		& \mathcal{L}V=-\psi_{1}{\rm Tr}\{C\varphi(\boldsymbol{\Upsilon}(\tilde{R}))^{\top}\Gamma_{c}\}+\psi_{1}\alpha_{b}+\label{eq:Attit_Lyap_2}\\
		& \frac{\psi_{2}}{4}{\rm Tr}\{(W_{\sigma}^{\top}\varphi(\boldsymbol{\Upsilon}(\tilde{R}))+\alpha_{\sigma})(W_{\sigma}^{\top}\varphi(\boldsymbol{\Upsilon}(\tilde{R}))+\alpha_{\sigma})^{\top}\}\nonumber \\
		& -{\rm Tr}\{\tilde{W}_{\sigma}^{\top}\Gamma_{\sigma}^{-1}\dot{\hat{W}}_{\sigma}\}\nonumber 
		\end{align}
		According to \textit{Young's inequality,} $\alpha_{\sigma}^{\top}W_{\sigma}^{\top}\varphi(\boldsymbol{\Upsilon}(\tilde{R}))\leq\frac{1}{2}\varphi(\boldsymbol{\Upsilon}(\tilde{R}))^{\top}\overline{W}_{\sigma}\varphi(\boldsymbol{\Upsilon}(\tilde{R}))+\frac{1}{2}||\alpha_{\sigma}||$.
		Therefore, one obtains
		\begin{align}
		& \mathcal{L}V\leq-\psi_{1}{\rm Tr}\{C\varphi(\boldsymbol{\Upsilon}(\tilde{R}))^{\top}\Gamma_{c}\}-{\rm Tr}\{\tilde{W}_{\sigma}^{\top}\Gamma_{\sigma}^{-1}\dot{\hat{W}}_{\sigma}\}\nonumber \\
		& \hspace{1em}+\frac{\psi_{2}}{2}\varphi(\boldsymbol{\Upsilon}(\tilde{R}))\varphi(\boldsymbol{\Upsilon}(\tilde{R}))^{\top}\overline{W}_{\sigma}+\psi_{1}\alpha_{b}+\frac{\psi_{2}}{2}||\alpha_{\sigma}||^{2}\label{eq:Attit_Lyap_3}
		\end{align}
		Note that $\psi_{1}\leq\exp(||\tilde{R}||_{{\rm I}})<3$ and $\psi_{2}\leq3\exp(||\tilde{R}||_{{\rm I}})<9$.
		In view of \eqref{eq:Attit_We}, let us replace $\overline{W}_{\sigma}$
		in \eqref{eq:NAV_Filter1_Detailed} by $\overline{W}_{\sigma}=\tilde{W}_{\sigma}+\hat{W}_{\sigma}$.
		Thus, using $\dot{\hat{W}}_{\sigma}$ and $C$ in \eqref{eq:NAV_Filter1_Detailed},
		the expression \eqref{eq:Attit_Lyap_3} can be reformulated in an
		inequality form as follows:
		\begin{align}
		\mathcal{L}V\leq & -\psi_{1}||\Gamma_{c}^{\top}\varphi(\boldsymbol{\Upsilon}(\tilde{R}))||^{2}-k_{\sigma}||\tilde{W}_{\sigma}||_{F}^{2}\nonumber \\
		& +k_{\sigma}||\tilde{W}_{\sigma}||_{F}||\overline{W}_{\sigma}||_{F}+3\alpha_{b}+\frac{9}{2}||\alpha_{\sigma}||^{2}\label{eq:Attit_Lyap_4}
		\end{align}
		Based on \textit{Young's inequality,} $k_{\sigma}||\tilde{W}_{\sigma}||_{F}||\overline{W}_{\sigma}||_{F}\leq\frac{k_{\sigma}}{2}||\tilde{W}_{\sigma}||_{F}^{2}+\frac{k_{\sigma}}{2}||\overline{W}_{\sigma}||_{F}^{2}$.
		Consider a hyperbolic tangent activation function $\varphi(a)=\frac{\exp(a)-\exp(-a)}{\exp(a)+\exp(-a)}$
		where $a\in\mathbb{R}$. One finds that $4||\Gamma_{c}^{\top}\varphi(\boldsymbol{\Upsilon}(\tilde{R}))||^{2}\geq k_{c}||\boldsymbol{\Upsilon}(\tilde{R})||^{2}$
		where $k_{c}=\underline{\lambda}(\Gamma_{c}^{\top}\Gamma_{c})$. Hence,
		for a hyperbolic tangent activation function one has
		\begin{align}
		\mathcal{L}V\leq & -\frac{k_{c}}{4}||\boldsymbol{\Upsilon}(\tilde{R})||^{2}-\frac{k_{b}}{2}||\tilde{W}_{\sigma}||_{F}^{2}+\eta\label{eq:Attit_Lyap_5}
		\end{align}
		where $\eta=\sup_{t\geq0}\frac{k_{b}}{2}||\overline{W}_{\sigma}||_{F}^{2}+3\alpha_{b}+\frac{9}{2}||\alpha_{\sigma}||^{2}$.
		This shows that $\mathcal{L}V$ is ultimately bounded. Let $\underline{\delta}\geq1-||\tilde{R}(0)||_{{\rm I}}$
		and recall Lemma \ref{lem:Lem1}. Accordingly, one shows
		\begin{align}
		\mathcal{L}V\leq & -e^{\top}\underbrace{\left[\begin{array}{cc}
			\underline{\delta}\,k_{c} & 0\\
			0 & k_{\sigma}
			\end{array}\right]}_{H_{3}}e+\eta\nonumber \\
		\leq & -\underline{\lambda}(H_{3})||e||^{2}+\eta\label{eq:Attit_Lyap_6}
		\end{align}
		where $e=[\sqrt{||\tilde{R}||_{{\rm I}}},||\tilde{W}_{\sigma}||_{F}]^{\top}$.
		Since $k_{\sigma}>0$ and $k_{c}>0$ and given that $||\tilde{R}(0)||_{{\rm I}}$
		does not belong to the unstable equilibria, it becomes apparent that
		$\underline{\lambda}(H_{3})>0$. Hence, $\mathcal{L}V<0$ if
		\[
		||e||^{2}>\frac{\eta}{\underline{\lambda}(H_{3})}
		\]
		Consequently, one finds
		\begin{equation}
		\frac{d\mathbb{E}[V]}{dt}=\mathbb{E}[\mathcal{L}V]\leq-\frac{\underline{\lambda}(H_{3})}{\overline{\lambda}(H_{2})}\mathbb{E}[V]+\eta\label{eq:Attit_Lyap7}
		\end{equation}
		Let us define $\beta=\frac{\underline{\lambda}(H_{3})}{\overline{\lambda}(H_{2})}$.
		Therefore, one obtains
		\begin{align}
		0\leq V(t) & \leq V(0)\exp(-\beta t)+\frac{\eta}{\beta}(1-\exp(-\beta t))\label{eq:Attit_Lyap8}
		\end{align}
		As such, it becomes apparent that $e$ is almost SGUUB which completes
		the proof.\end{proof}
	
	The comprehensive steps of the neural-adaptive stochastic attitude
	filter implementation in its discrete form are listed in Algorithm
	\ref{alg:DiscFilter1} with $\Delta t$ being a small sampling time.
	Singular value decomposition \cite{markley1988attitude} has been
	utilized a method of attitude reconstruction. In Algorithm \ref{alg:DiscFilter1},
	$s_{i}$ denotes $i$th sensor measurement confidence level with $\sum_{i=1}^{N}s_{i}=1$.
	\begin{algorithm}
		\caption{\label{alg:DiscFilter1}Neural-adaptive stochastic attitude estimator}
		
		\textbf{Initialization}:
		\begin{enumerate}
			\item[{\footnotesize{}1:}] Set $\hat{R}[0]=\hat{R}_{0}\in\mathbb{SO}\left(3\right)$, $\hat{W}_{\sigma}[0]=\hat{W}_{\sigma|0}=0_{q\times q}$,
			$q>0$, $s_{i}\geq0$ for all $i\geq2$, select $\Gamma_{\sigma},k_{\sigma}>0$,
			$\underline{\lambda}(\Gamma_{c}^{\top}\Gamma_{c})>0$, and set $k=0$.\vspace{1mm}
		\end{enumerate}
		\textbf{while}
		\begin{enumerate}
			\item[] \textcolor{blue}{/{*} Attitude reconstruction using Singular Value
				Decomposition {*}/}\vspace{1mm}
			\item[{\footnotesize{}2:}] $\begin{cases}
			{\bf r}_{i} & =\frac{r_{i}}{||r_{i}||},\hspace{1em}{\bf y}_{i}=\frac{y_{i}}{||y_{i}||},\hspace{1em}i=1,2,\ldots,N\\
			B & =\sum_{i=1}^{n}s_{i}{\bf y}_{i}{\bf r}_{i}^{\top}=USV^{\top}\\
			U_{+} & =U\cdot diag(1,1,\det(U))\\
			V_{+} & =V\cdot diag(1,1,\det(V))\\
			R_{y} & =V_{+}U_{+}^{\top}
			\end{cases}$\vspace{1mm}
			\item[{\footnotesize{}3:}] $\tilde{R}_{k}=R_{y}^{\top}\hat{R}_{k}$ and $\boldsymbol{\Upsilon}=\boldsymbol{\Upsilon}(\tilde{R}_{k})=\mathbf{vex}(\boldsymbol{\mathcal{P}}_{a}(\tilde{R}))$\vspace{1mm}
			\item[{\footnotesize{}4:}] $\varphi(\boldsymbol{\Upsilon})=\frac{\exp(\boldsymbol{\Upsilon})-\exp(-\boldsymbol{\Upsilon})}{\exp(\boldsymbol{\Upsilon})+\exp(-\boldsymbol{\Upsilon})}$\textcolor{blue}{{}
				/{*} hyperbolic tangent activation function {*}/}\vspace{1mm}
			\item[{\footnotesize{}5:}] $\hat{W}_{\sigma|k}=\hat{W}_{\sigma|k-1}+\Delta t\Gamma_{\sigma}(\psi_{2}\varphi(\boldsymbol{\Upsilon})\varphi(\boldsymbol{\Upsilon})^{\top}-k_{\sigma}\hat{W}_{\sigma|k-1})$\vspace{1mm}
			\item[{\footnotesize{}6:}] $C=\left(\Gamma_{c}^{\top}+\frac{\psi_{2}}{2\psi_{1}}(\Gamma_{c}^{\top}\Gamma_{c})^{-1}\Gamma_{c}^{\top}\hat{W}_{\sigma|k}\right)\varphi(\boldsymbol{\Upsilon})$\vspace{1mm}
			\item[] \textcolor{blue}{/{*} angle-axis parameterization {*}/}
			\item[{\footnotesize{}7:}] $\begin{cases}
			\varrho & =(\Omega_{m|k}-C)\Delta t\\
			\mu & =||\varrho||,\hspace{1em}x=\varrho/||\varrho||\\
			\mathcal{R}_{exp} & =\mathbf{I}_{3}+\sin(\mu)[x]_{\times}+(1-\cos(\mu))[x]_{\times}^{2}
			\end{cases}$ \vspace{1mm}
			\item[{\footnotesize{}8:}] $\hat{R}_{k+1}=\hat{R}_{k}\mathcal{R}_{exp}$\vspace{1mm}
			\item[{\footnotesize{}9:}] $k+1\rightarrow k$
		\end{enumerate}
		\textbf{end while}
	\end{algorithm}

	\section{Simulation Results \label{sec:SE3_Simulations}}
	
	This section illustrates the functionality of the proposed neural-adaptive
	stochastic filter on the Lie group of $\mathbb{SO}\left(3\right)$.
	The discrete filter presented in Algorithm \ref{alg:DiscFilter1}
	has been tested at a sampling rate of $\Delta t=0.01$ seconds. Assume
	that the initial value of $R$ is $R(0)=\mathbf{I}_{3}\in\mathbb{SO}\left(3\right)$
	and the true angular velocity be as below:
	\[
	\Omega=0.6\left[\sin(0.4t),\sin(0.7t+\frac{\pi}{4}),0.4\cos(0.3t)\right]^{\top},\,\text{(rad/sec)}
	\]
	Let the true angular velocity be attached with unknown normally distributed
	random noise $n=\mathcal{N}(0,0.11)$ (rad/sec) (zero mean and standard
	deviation of $0.11$), see \eqref{eq:Attit_Om_m}. Define two observations
	in $\{\mathcal{I}\}$: $r_{1}=[1,-1,1]^{\top}$ and $r_{2}=[0,0,1]^{\top}$.
	Let $\{\mathcal{B}\}$ measurements be corrupted with unknown normally
	distributed random noise $n_{1}=n_{2}=\mathcal{N}(0,0.1)$, see \eqref{eq:Attit_Vec_yi}.
	Let us consider three neurons ($q=3$). Consider selecting the design
	parameters as follows: $\Gamma_{c}=2\mathbf{I}_{3}$, $\Gamma_{\sigma}=2\mathbf{I}_{3}$,
	and $k_{\sigma}=1$. Let the initial estimate of neural network weights
	be set to $\hat{W}(0)=0_{3\times3}$ and the initial estimate of the
	attitude be
	\[
	\hat{R}(0)=\left[\begin{array}{ccc}
	-0.9214 & -0.0103 & 0.3884\\
	0.2753 & -0.7227 & 0.634\\
	0.2742 & 0.6911 & 0.6687
	\end{array}\right]\in\mathbb{SO}\left(3\right)
	\]
	where $||\tilde{R}(0)||_{{\rm I}}=\frac{1}{4}{\rm Tr}\{\mathbf{I}_{3}-R_{0}^{\top}\hat{R}_{0}\}\approx0.994$
	approaching the unstable equilibrium $+1$. As to activation function,
	we selected a hyperbolic tangent activation function:
	\[
	\varphi(\alpha)=\frac{\exp(\alpha)-\exp(-\alpha)}{\exp(\alpha)+\exp(-\alpha)},\hspace{1em}\alpha\in\mathbb{R}
	\]
	
	Fig. \ref{fig:Fig_Noise} illustrates the high level of noise corrupting
	the angular velocity measurements in comparison to the true data.
	In Fig. \ref{fig:Fig_Euler}, the estimated Euler angles (roll ($\hat{\phi}$),
	pitch ($\hat{\theta}$), and yaw ($\hat{\psi}$)) are plotted against
	the true Euler angles ($\phi$, $\theta$, $\psi$). Fig. \ref{fig:Fig_Euler}
	demonstrates fast and strong tracking capability of the proposed approach.
	The effectiveness and robustness of the neural-adaptive approach are
	illustrated in Fig. \ref{fig:Fig_ERR} where the error initiates at
	a large value and rapidly reaches close neighborhood of the origin.
	Table \ref{tab:SO3_1} shows statistical analysis of mean and standard
	deviation (std) of the steady-state error values between 5 to 29 seconds
	with respect to the number of neurons. As illustrated by Table \ref{tab:SO3_1},
	greater number of neurons results in improved steady-state error convergence.
	Finally, Fig. \ref{fig:Fig_West} depicts the boundedness of the neural-adaptive
	estimates as they converge close to zero as $||\tilde{R}||_{{\rm I}}\rightarrow0$.
	
	\begin{figure}[h]
		\centering{}\includegraphics[scale=0.32]{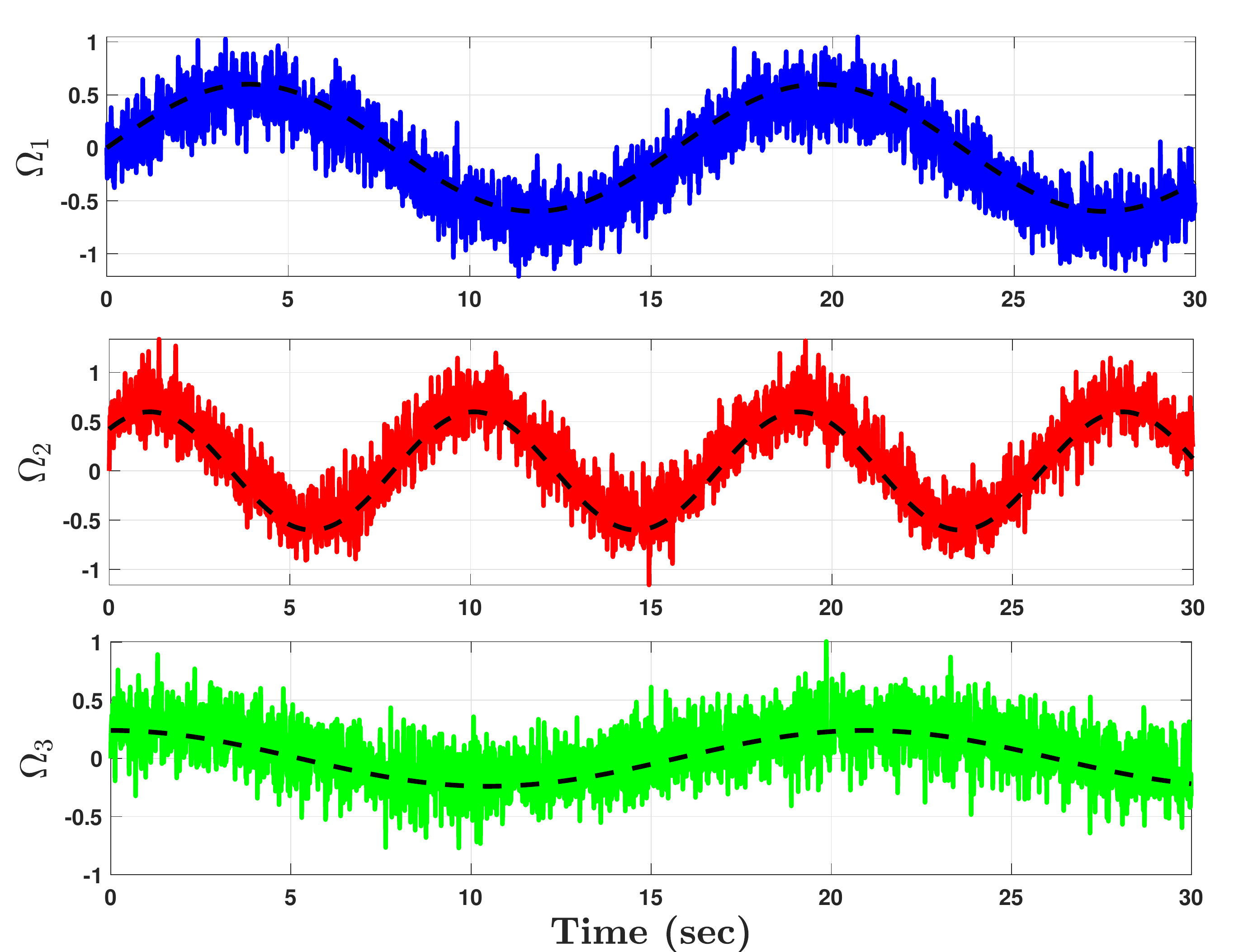}\caption{Rate gyro: True (black center-line) and measurements (colored)}
		\label{fig:Fig_Noise}
	\end{figure}
	
	\begin{figure}[h]
		\centering{}\includegraphics[scale=0.3]{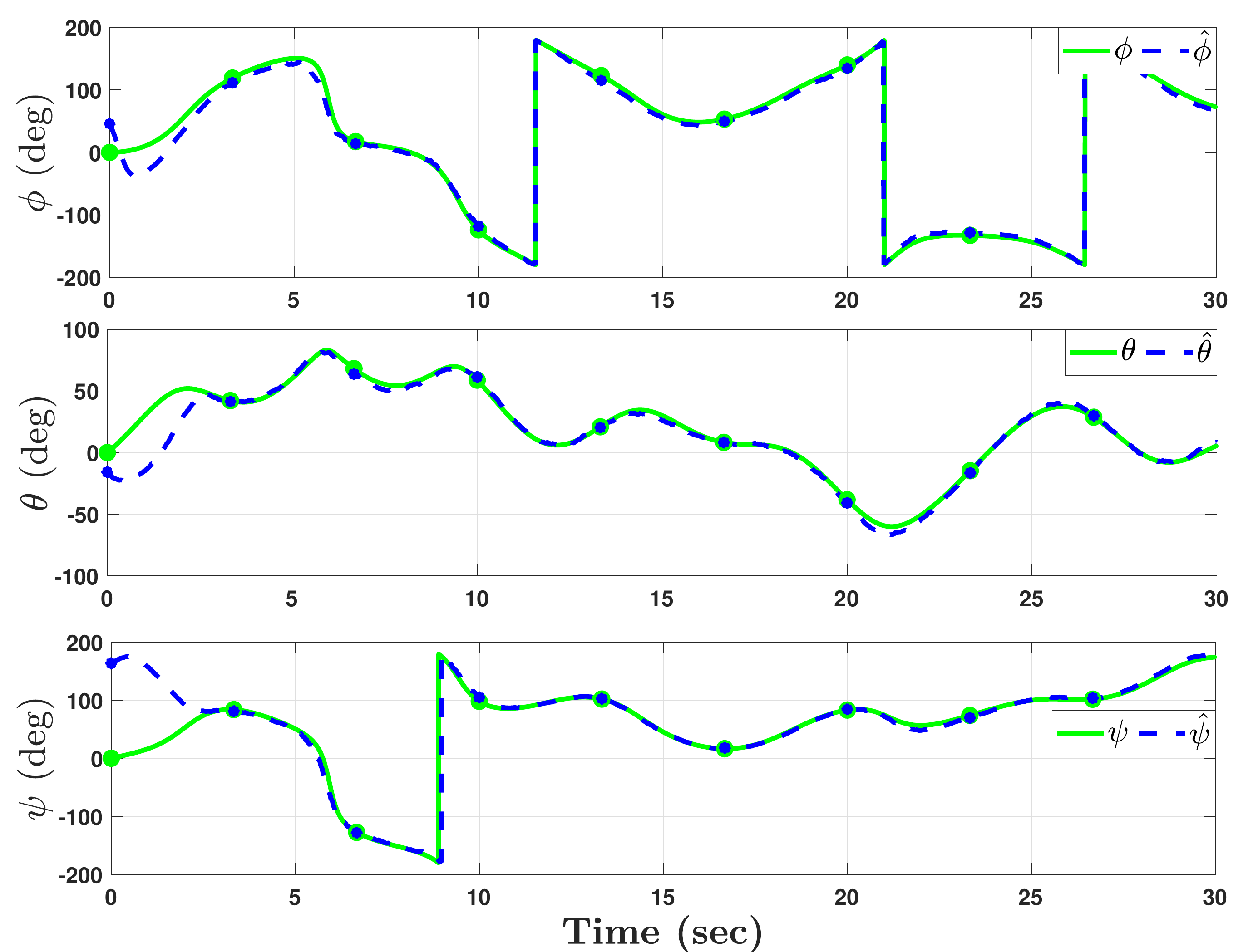}\caption{Euler angles: True (green solid-line) and estimated (blue dash-line)
			using 3 neurons}
		\label{fig:Fig_Euler}
	\end{figure}
	
	\begin{figure}[h]
		\centering{}\includegraphics[scale=0.26]{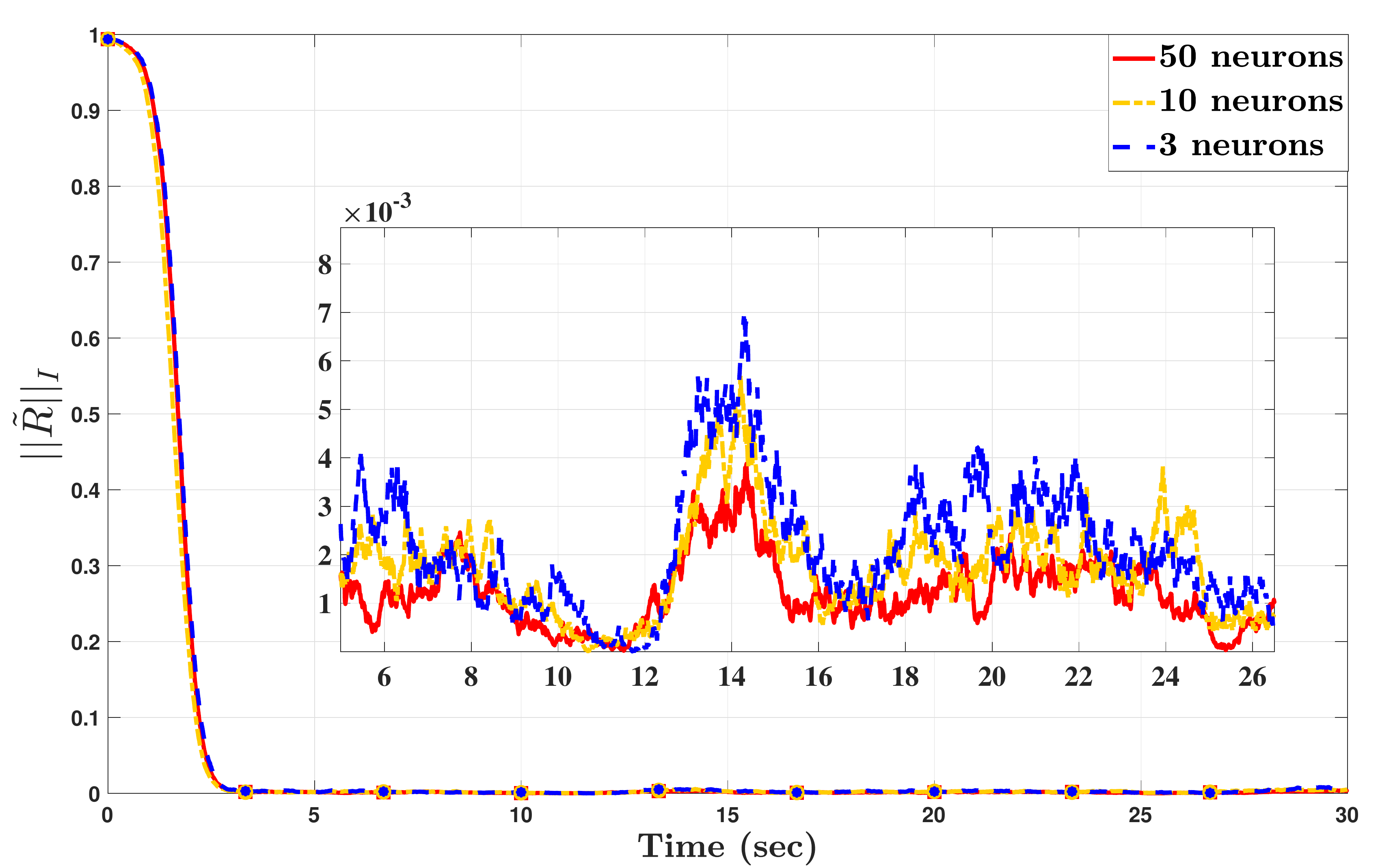}\caption{Normalized Euclidean error $||\tilde{R}||_{{\rm I}}=\frac{1}{4}{\rm Tr}\{\mathbf{I}_{3}-R_{k}^{\top}\hat{R}_{k}\}$.}
		\label{fig:Fig_ERR}
	\end{figure}
	
	\begin{table}[H]
		\caption{\label{tab:SO3_1}Statistical analysis of the steady-state error with
			respect to the number of neurons.}
		
		\centering{}%
		\begin{tabular}{c|>{\centering}p{1.8cm}|>{\centering}p{1.8cm}|>{\centering}p{1.8cm}}
			\hline 
			\noalign{\vskip\doublerulesep}
			\multicolumn{4}{c}{Output data of $||\tilde{R}||_{{\rm I}}=\frac{1}{4}{\rm Tr}\{\mathbf{I}_{3}-R_{k}^{\top}\hat{R}_{k}\}$
				over the period (5-29 sec)}\tabularnewline[\doublerulesep]
			\hline 
			\hline 
			\noalign{\vskip\doublerulesep}
			Neurons number & 3 & 10 & 50\tabularnewline[\doublerulesep]
			\hline 
			\noalign{\vskip\doublerulesep}
			Mean & $2.3\times10^{-3}$ & $2\times10^{-3}$ & $1.4\times10^{-3}$\tabularnewline[\doublerulesep]
			\hline 
			\noalign{\vskip\doublerulesep}
			STD & $1.9\times10^{-3}$ & $1.4\times10^{-3}$ & $9\times10^{-4}$\tabularnewline[\doublerulesep]
			\hline 
		\end{tabular}
	\end{table}
	
	\begin{figure}[h]
		\centering{}\includegraphics[scale=0.32]{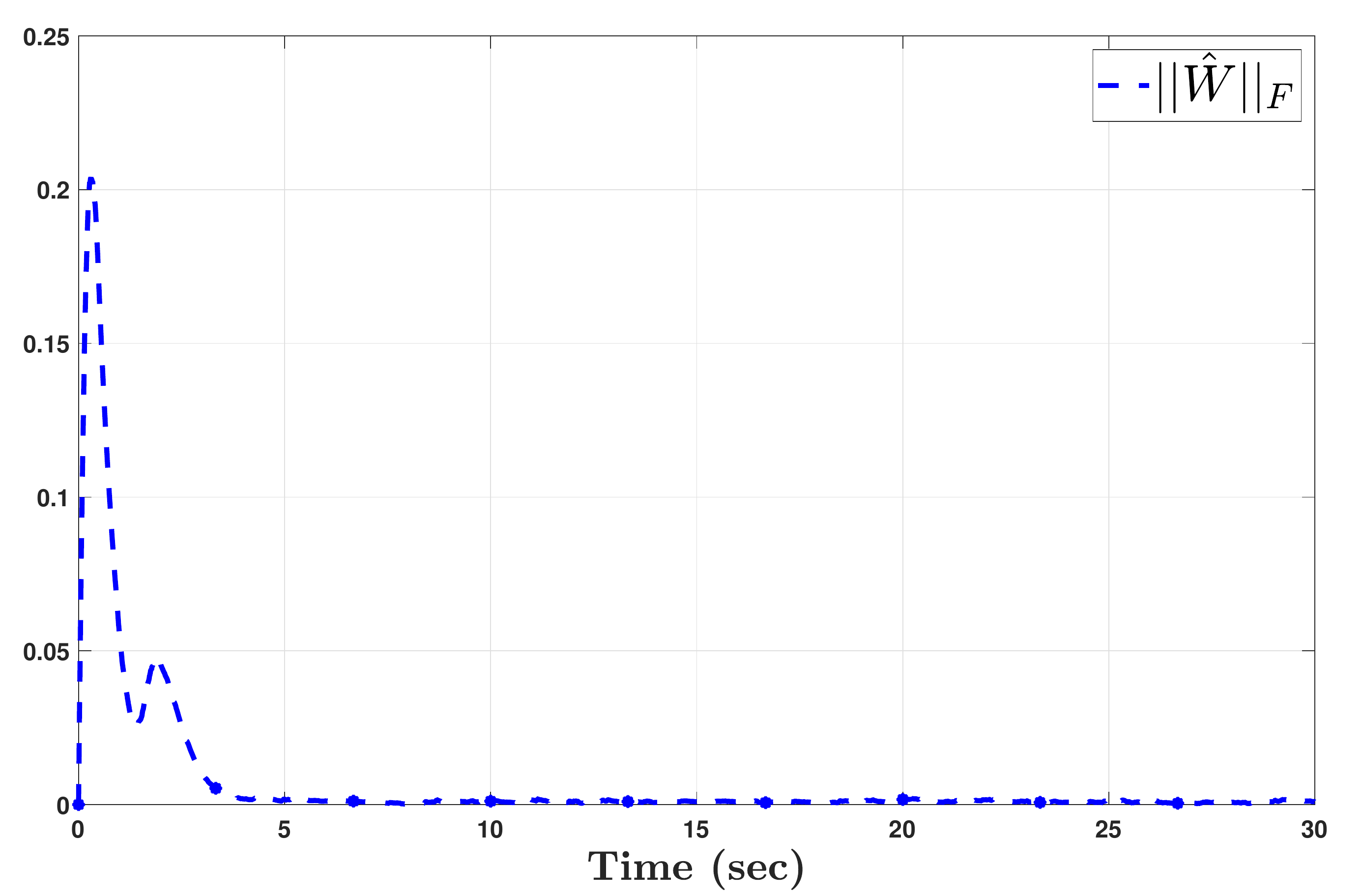}\caption{Frobenius norm of neural-adaptive estimates (3 neurons).}
		\label{fig:Fig_West}
	\end{figure}

	\section{Conclusion \label{sec:SE3_Conclusion}}
	
	Accurate attitude estimation is a fundamental component of successful
	robotic applications. The estimation can be achieved using a group
	of observations and measurements. Accurate estimation become challenging
	when low-cost measurement units are utilized. This work addressed
	the attitude estimation problem using a neural-adaptive stochastic
	filter on the Special Orthogonal Group $\mathbb{SO}(3)$. The novel
	filter accounts for the noise present in the gyroscope measurements.
	The proposed filter is ensured to be almost SGUUB in the mean square.
	The numerical simulation illustrates robustness and rapid adaptability
	of the proposed neural-adaptive approach. 
	
	\section*{Acknowledgment}
	
	The author would like to thank \textbf{Maria Shaposhnikova} for proofreading
	the article.

	\subsection*{Appendix\label{subsec:Appendix-A}}
	\begin{center}
		\textbf{Neural-adaptive Filter Quaternion Representation}
		\par\end{center}
	
	\noindent Let $\mathbb{S}^{3}=\{\left.Q\in\mathbb{R}^{4}\right|||Q||=\sqrt{q_{0}^{2}+q^{\top}q}=1\}$
	and let $Q=[q_{0},q^{\top}]^{\top}\in\mathbb{S}^{3}$ be a unit-quaternion
	vector with $q_{0}\in\mathbb{R}$ and $q\in\mathbb{R}^{3}$. Let $Q^{-1}=[\begin{array}{cc}
	q_{0} & -q^{\top}\end{array}]^{\top}\in\mathbb{S}^{3}$ be the inverse of $Q\in\mathbb{S}^{3}$. Consider $\odot$ to be
	a quaternion product. Then, for $Q_{1}=[\begin{array}{cc}
	q_{01} & q_{1}^{\top}\end{array}]^{\top}\in\mathbb{S}^{3}$ and $Q_{2}=[\begin{array}{cc}
	q_{02} & q_{2}^{\top}\end{array}]^{\top}\in\mathbb{S}^{3}$, one has
	\[
	Q_{1}\odot Q_{2}=\left[\begin{array}{c}
	q_{01}q_{02}-q_{1}^{\top}q_{2}\\
	q_{01}q_{2}+q_{02}q_{1}+[q_{1}]_{\times}q_{2}
	\end{array}\right]
	\]
	$\mathbb{S}^{3}$ can be mapped to $\mathbb{SO}\left(3\right)$ as
	below \cite{hashim2019AtiitudeSurvey,shuster1993survey}
	\begin{align}
	\mathcal{R}_{Q} & =(q_{0}^{2}-||q||^{2})\mathbf{I}_{3}+2qq^{\top}+2q_{0}\left[q\right]_{\times}\in\mathbb{SO}\left(3\right)\label{eq:NAV_Append_SO3}
	\end{align}
	Let $Q_{y}$ be the reconstructed attitude, obtained for instance,
	using QUEST \cite{shuster1981three}. Define $\hat{Q}=[\hat{q}_{0},\hat{q}^{\top}]^{\top}\in\mathbb{S}^{3}$
	as the estimate of $Q=[q_{0},q^{\top}]^{\top}\in\mathbb{S}^{3}$,
	and let the error in estimation be $\tilde{Q}=Q_{y}^{-1}\odot\hat{Q}=[\tilde{q}_{0},\tilde{q}^{\top}]^{\top}\in\mathbb{S}^{3}$.
	The quaternion representation of the neural-adaptive stochastic attitude
	filter in \eqref{eq:NAV_Filter1_Detailed} is as below:
	\begin{equation}
	\begin{cases}
	\dot{\hat{W}}_{\sigma} & =\frac{\psi_{2}}{2}\Gamma_{\sigma}\varphi(2\tilde{q}_{0}\tilde{q})\varphi(2\tilde{q}_{0}\tilde{q})^{\top}-k_{\sigma}\Gamma_{\sigma}\hat{W}_{\sigma}\\
	C & =\left(\Gamma_{c}^{\top}+\frac{\psi_{2}}{2\psi_{1}}(\Gamma_{c}^{\top}\Gamma_{c})^{-1}\Gamma_{c}^{\top}\hat{W}_{\sigma}\right)\varphi(2\tilde{q}_{0}\tilde{q})\\
	u & =\Omega_{m}-C\\
	\Phi & =\left[\begin{array}{cc}
	0 & -u^{\top}\\
	u & -[u]_{\times}
	\end{array}\right]\\
	\dot{\hat{Q}} & =\frac{1}{2}\Phi\hat{Q}
	\end{cases}\label{eq:Quat}
	\end{equation}
	where $\boldsymbol{\Upsilon}(\tilde{\mathcal{R}}_{Q})=2\tilde{q}_{0}\tilde{q}$,
	$||\tilde{\mathcal{R}}_{Q}||_{{\rm I}}=1-\tilde{q}_{0}^{2}$, $\psi_{1}=\frac{1}{2}(1+||\tilde{\mathcal{R}}_{Q}||_{{\rm I}})\exp(||\tilde{\mathcal{R}}_{Q}||_{{\rm I}})$,
	and $\psi_{2}=\frac{1}{2}(2+||\tilde{\mathcal{R}}_{Q}||_{{\rm I}})\exp(||\tilde{\mathcal{R}}_{Q}||_{{\rm I}})$.

	\balance
	\bibliographystyle{IEEEtran}
	\bibliography{bib_Neuro_Attit}
\end{document}